%% file: main.tex
\documentclass[submission,copyright]{eptcs}
\providecommand{\event}{ICE 2019} % Name of the event you are submitting to

\input{macros}
\input{styles}
\input{macros_util}

\usepackage{algorithm}
\usepackage[end]{algpseudocode}

\title{On Learning Nominal Automata with Binders}
\author{Yi Xiao
  \institute{Department of Informatics\\
    University of Leicester (UK)
  }
  \email{yx104@leicester.ac.uk}
  \and
  Emilio Tuosto
  \institute{Gran Sasso Science Institute (IT) and
  \\ Department of Informatics, University of Leicester (UK)}
\email{emilio.tuosto@gssi.it}
}
\def\titlerunning{On Learning Nominal Automata with Binders}
\def\authorrunning{Y. Xiao \& E. Tuosto}
\begin{document}
\maketitle
\begin{abstract}
  We investigate a learning algorithm in the context of \emph{nominal
    automata}, an extension of classical automata to alphabets
  featuring names.
  This class of automata captures \emph{nominal regular languages};
  analogously to the classical language theory, nominal automata have
  been shown to characterise \emph{nominal regular expressions with
    binders}. These formalisms are amenable to abstract modelling
  resource-aware computations.
  
  We propose a learning algorithm on nominal regular languages with
  binders. Our algorithm generalises Angluin's $L^\star$ algorithm
  with respect to nominal regular languages with binders. We show the
  correctness and study the theoretical complexity of our algorithm.
\end{abstract}

\section{Introduction}\label{sec:intro}
\input{intro}

\section{Background}\label{sec:bkg}
\input{bkg}

\section{Nominal Learning}\label{st:LNL}\label{sec:nla}
\input{nla}

\section{Conclusions}\label{sec:conc}
\input{conc}

\bibliographystyle{eptcs}
\bibliography{bib}

\newpage

\appendix
% \section{Running \lstar}\label{sec:AngluinAlgorithm}
% \input{rl}

% \section{The symbol \quo{P}}
% \input{rnl}

\end{document}

%% file: macros.tex
\newcommand{\lstar}{\mbox{$\textsf{L}^\star$}\xspace}%Angluin's alogrithm
\newcommand{\nlstar}{\mbox{$\textsf{nL}^\star$}\xspace}
\newcommand{\emptystring}{\epsilon}
\renewcommand{\ll}{\langle\!\!\langle}
\renewcommand{\gg}{\rangle\!\!\rangle}
\newcommand{\length}[1]{|#1|}
\newcommand{\languageof}[1]{\mathcal{L}(#1)} % the function that get a languge of a regular expression
\newcommand{\lofm}[1]{\mathfrak{L}(#1)} % the function that get a languge of a regular expression

\newcommand{\row}{\textit{row}}%function of row in a table
\newcommand{\ce}[2]{\chi(#2,#1)}%the canonical representation of a nominal regular expression
\newcommand{\reg}[1]{\|#1\|}%represent \|1\|
\newcommand{\Mreg}[1]{\|#1\|_M}%represent \|1\|_M
\newcommand{\Otable}{(S,E,T,{\Alp_n})}%new n-observation table symbol presentation
\newcommand{\nOtable}[1][{}]{n-observation  table{#1}\xspace}%n-observation table text presentation
\newcommand{\labelsinT}{\{0, 1, P, \bot\}} %labels in observation table
\newcommand{\state}[1]{(\row(#1), \reg{#1})} %states associated to a table
\newcommand{\Nom}{\mathcal{N}}
\newcommand{\Nat}{\mathbb{N}}
\newcommand{\Alp}{\mathbb{A}}
\newcommand{\Nf}{\Nom_\mathit{fin}}
\newcommand{\aname}[1][n]{\mathsf{#1}}

%% file: styles.tex
\usepackage[utf8]{inputenc}
\usepackage{xargs}
\usepackage{xspace}

\usepackage[nomargin,multiuser,inline,draft]{fixme} 
\fxusetheme{color}
\usepackage{rotating}

\usepackage{caption}                % Needed by subcaption
\usepackage{cmll} % \bigwith symbol
\usepackage{colonequals}            % More symbols (colon equalities etc.)
\usepackage{environ}                % Provides NewEnviron w/ eager expansion
\usepackage{etex,etoolbox}          % Needed by theorem-appendix.tex
\usepackage{graphicx}               % For \scalebox etc.

\usepackage[protrusion=true,expansion=true]{microtype} % Better typography
\usepackage{multirow}               % Table multirow/multicolumn support
\usepackage{nicefrac}
\usepackage{proof}                  % Logic proofs
\usepackage{bussproofs}
\EnableBpAbbreviations

\usepackage{rotating}               % For \sidewaysfigure etc.
\usepackage{subcaption}             % Support for subfigures
\usepackage{thmtools, thm-restate, thm-autoref}  % Advanced theorem handling
\usepackage{xifthen}                % for conditional commands
\usepackage{appendix}               % for correct numeration of statements

\usepackage{stmaryrd}
 \usepackage{pdfsync}

\usepackage[usenames,dvipsnames,svgnames,table]{xcolor} % must be loaded first
\usepackage{tikz}                   % Graphs
\usetikzlibrary{shadows,arrows,shapes,automata,positioning,decorations.markings}

\usepackage{amsmath,amssymb,amsthm,amsbsy}
\usepackage{mathtools} % \xRightArrow and stuff - load late to avoid conflicts
\newtheorem{definition}{Definition}
\newtheorem{theorem}{Theorem}
\newtheorem{lemma}{Lemma}
\newtheorem{example}{Example}

\usepackage{hyperref}
\usepackage[capitalise]{cleveref}
\DeclareMathSizes{10}{10}{6}{6}

\usepackage{csvsimple}

%% file: macros_util.tex
\usepackage{xargs}
\newif\ifemi
\DeclareGraphicsExtensions{%
    .png,.PNG,%
    .pdf,.PDF,%
    .jpg,.mps,.jpeg,.jbig2,.jb2,.JPG,.JPEG,.JBIG2,.JB2}

\usepackage[UKenglish]{babel}
  
\usepackage{fixme}
\fxusetheme{color}
\FXRegisterAuthor{eM}{aeM}{\color{orange} {\underline{eM}}}

\usepackage[normalem]{ulem} % underline command breaks over line ends

\usepackage{xifthen}        % for conditional commands
\newcommand{\ifempty}[3]{%
  \ifthenelse{\isempty{#1}}{#2}{#3}%
}

\newcommand{\hidden}[1]{}
\newcommand{\hide}[1]{}

\newcommand{\cf}[2]{
  \fontsize{#1}{#1}{\selectfont{#2}}
}
\ifemi
\usepackage{showlabels}

\newcommand{\emi}[2]{
  \marginpar{\fcolorbox{red}{shadecolor}{\cf{#1}{{#2}}}}
}
\newcommand{\emic}[2]{\par
  \fcolorbox{red}{shadecolor}{\parbox{\linewidth}{ 
      \color{gray}
      \begin{description}
      \item[{\color{blue} #2}]{\sf #1}
      \end{description}}}
}
\else
\newcommand{\emi}[2]{}
\newcommand{\emic}[2]{}{}
\fi

\newcommand{\sst}{\;\big|\;}
\newcommand{\qst}{\;\colon\;} %such that

\newcommand{\emptyword}{\varepsilon}

{\bfseries}{\rmfamily}
{\bfseries}{\rmfamily}

\newcommand{\quo}[1]{\lq\lq {#1}\rq\rq}
\def\finex{{\unskip\nobreak\hfil
\penalty50\hskip1em\null\nobreak\hfil$\diamond$
\parfillskip=0pt\finalhyphendemerits=0\endgraf}}

\newcommand{\grmeq}{\; ::= \;}
\newcommand{\grmor}{\; \big| \;}

\definecolor{shadecolor}{rgb}{1,0.99,0.9}
\definecolor{bg}{rgb}{0.95,0.95,0.95}

%% file: intro.tex
This paper combines \emph{nominal languages} and \emph{learning
  automata} to abstractly model computations connected with
\emph{resource awareness}.
Here, we do not restrict ourselves to a specific type of resources;
rather we think of resources in a very abstract and general sense.
We use \emph{names} as models of resources and (abstract) operations
on names as developed in \emph{nominal} languages (see \cref{sec:bkg}
for an overview) as mechanisms to capture basic properties of
resources; in particular we focus on the \emph{dynamic allocation and
  deallocation} of resources.
More precisely, we take inspiration from binders with dynamic scoping
of nominal languages in an operational context based on finite state
\emph{nominal} automata.
The states of these automata have transitions to explicitly (i)
allocate names, corresponding to \emph{scope extrusion} of nominal
languages, and (ii) to deallocate names corresponding to \emph{garbage
  collection} of (unused) names.  Our theory sets in the context of
\emph{nominal regular expressions} that transfer the traditional
Kleene theorem to the nominal framework adopted here.
In fact, the class of nominal languages that we consider can be
characterised as those accepted by nominal automata or,
equivalently, that can be generated by \emph{nominal regular
  expressions}.
The latter algebraic presentation (that we borrow from the literature
and review in \cref{sec:bkg}) features, besides the usual
operations of regular expressions (union, concatenation, Kleene-star),
a name binding mechanism and a special \emph{resource-aware}
complementation operation.
Our results rely on the closure properties of these class of languages
that has been already demonstrated in the literature.

In this context, we develop a learning algorithm for nominal automata.
We take inspiration from the \lstar algorithm of Angluin (also
reviewed in \cref{sec:bkg}).
As we will see, the design of the algorithm requires some ingenuity
and opens up the possibility of interesting investigations due to
richer structure brought in by names and name binding.

%\section{Literature Review}

\paragraph{Nominal languages and learning}

% \begin{old}
% The theory of nominal languages has been advocated as a suitable
% abstraction for computations \emph{with resources} emerging from the
% so-called nominal calculi which bred after the seminal work
% introducing $\pi$-calculus~\cite{mpw92,mil99,sw02}.
% %
% Abstract theories capturing the computational phenomena in this
% context have been developed in~\cite{gp99,gp02,gabbay:thesis} in
% parallel with a theory of nominal automata~\cite{mp00,fmp02,pis99}.
% %
% The formal connections between these theories have been unveiled
% in~\cite{GadducciMM06}.
% %
% Later,~\cite{Kurz0T12} proposed the notion of
% nominal regular languages and the use of nominal automata as acceptors
% of such languages.
% \end{old}
%
The pioneering work on languages on infinite alphabet is~\cite{kami-fran:TCS94}.  And, 
the theory of nominal languages has been advocated as a suitable
abstraction for computations \emph{with resources} emerging from the
so-called nominal calculi which bred after the seminal work
introducing the $\pi$-calculus~\cite{mpw92,mil99,sw02}.
Abstract theories capturing the computational phenomena in this
context have been developed in~\cite{gp99,gp02,gabbay:thesis} in
parallel with a theory of nominal automata~\cite{mp00,fmp02,pis99}.
The formal connections between these theories have been unveiled
in~\cite{GadducciMM06}.
Later,~\cite{Kurz0T12} proposed the notion of
nominal regular languages and the use of nominal automata as acceptors
of such languages.
As observed in~\cite{Kurz0T12,BojanczykKL14} are not suitable to
handle name binding as registers are 'global'; the nominal model
in~\cite{BojanczykKL14} is instead closer (see also the comment below
Example 11.4 of~\cite{BojanczykKL14}) to history dependent
automata~\cite{pis99}, which are also the inspiration for the model of
automata in ~\cite{Kurz0T12}. The nominal automata
in~\cite{BojanczykKL14} are (abstractions of) deterministic
HD-automata (which can be seen as 'implementation' of orbit-finite
nominal automata following the connection between nominal and named
set of~\cite{ GadducciMM06}).
This class of automata is more expressive than the classes of automata
capturing nominal regular expressions as ours or nominal
Kleene algebras~\cite{KozenMP015}.
In fact, as noted in~\cite{KozenMP015} this automata accept languages
with words having arbitrarely deep nesting of binders.
However, orbit-finite nominal automata are not closed under any
reasonable notion of complementation~\cite{BojanczykKL14}. Note that
the resource-sensitive complementation operation of~\cite{Kurz0T12}
is essential in our context. On the other hand, the use of symmetries
to capture binding offers a more flexible mechanism to express
patterns or words that escape the constraints that the use of 'nested
scoping' imposes in our language.

A main motivation for this proposal is the abstract characterisation
of basic features of computations \emph{with resources}.
For instance, nominal automata have applications to the verification
of protocols and systems~\cite{ffgmpr97,fgmpr98}.
Other approaches to verify resource-aware computations have also been
based on automata models~\cite{BartolettiDFZ09,dfm12,dfm13} employ
usage automata (UA) to express and model check patterns of
resource-usage.
A distinguishing feature of the approach in~\cite{Kurz0T12,Kurz0T2012}
is that allocation and deallocation of resources is abstracted away
with \emph{binders}.
Inspired by the \emph{scope extrusion} mechanism of the
$\pi$-calculus, the allocation of a resource corresponds to an
(explicit) operation that introduces a fresh name; likewise, the
deallocation of a resource corresponds to an (explicit) operation to
\quo{free} names.
We illustrate this idea with an example.
Consider the following expression
\begin{align*}
  \hat E & = \langle n. \langle m. m \rangle ^\star n \langle k. k^\star \rangle n \rangle
\end{align*}
which is a nominal regular expression where $n, m, k$ are names,
$\_^\star$ is the usual Kleene-star operation, and subexpressions of
the form $\langle n . E \rangle$ represent the binding mechanism
whereby name $n$ is bound (that is ``local'') to expression $E$.
Intuitively, $\hat E$ describes a language of words starting with the allocation of a freshly generated
name, conventionally denoted $n$, followed by the words generated by
the subexpression $\langle m. m \rangle ^\star$ post-fixed by
$n$, and so on.
Note that in $\hat E$ name $m$ occurs in a nested binder for name $n$.
According to~\cite{Kurz0T12}, $\hat E$ corresponds to the following
nominal automaton:
\[
  \begin{tikzpicture}[>=stealth',shorten >=1pt,auto,node distance=20mm]
    \node[initial,state,inner sep=1pt,minimum size=0pt] (q0)      {$q_0$};
    \node[state,inner sep=1pt,minimum size=0pt]         (q2) [above of=q0]  {$q_2$};
    \node[state,inner sep=1pt,minimum size=0pt]         (q3) [above of=q2]  {$q_3$};
    \node[state,inner sep=1pt,minimum size=0pt]         (q5) [right of = q2, xshift=1cm]  {$q_5$};
    \node[state,inner sep=1pt,minimum size=0pt]         (q6) [above of=q5]  {$q_6$};
    \node[state,inner sep=1pt,minimum size=0pt]         (q7) [right of=q5, xshift=1cm]  {$q_7$};
    \node[state,accepting,inner sep=1pt,minimum size=0pt]         (q4) [below of=q7]  {$q_4$};
    \path[->] (q0)  edge [bend left] node {$\ll$} (q2)	
    (q2) edge node {$1$} (q5)
    edge [ bend left] node {$\ll$} (q3)
    (q3) edge [loop above] node {$2$} (q3)
    edge[ bend left]  node {$\gg$} (q2)
    (q5) edge[bend left]  node {$\ll$} (q6)
    edge  node {$1$} (q7)
    (q6) edge [loop above] node {$2$} (q6)
    edge [bend left] node {$\gg$} (q5)
    (q7) edge [bend left] node {$\gg$} (q4);
  \end{tikzpicture}
\]
which from the initial state $q_0$ allocates a fresh name through the
transition labelled $\ll$ to state $q_2$.
Notice how bound names are rendered in the nominal automaton: they are
concretely represented as (strictly positive) natural numbers.
This allows us to abstract away from the identities of bound names.
In fact, the identity of bound names is immaterial and can be
\emph{alpha-converted}, that is replace with any other name provided
that the name has not been used already.
The use of numbers enables a simple ``implementation'' of
alpha-conversion.
More precisely, think of numbers as being addresses of registers
of states.
For instance, $q_3$ has 2 registers addressed by $1$ and $2$
respectively.
Then the self-loop transition in state $q_3$ can consume any name $n$,
provided that $n$ is different than the name (currently) stored in
register $1$.
Finally, note that the content of registers is local to states; once
a deallocation transition 
$\gg$ is fired, the content in last allocated register is disregarded.

For a practical example, we consider a scenario based on servers to
show how nominal regular languages with binders can suitably specify
usage policies of servers $S_1,\dots ,S_k$, such that,
$\forall 1\leq h\leq k\ S_h$ offers operations
$\{o_{h_1},\dots , o_{h_k}\}=O_h$. Given an alphabet
\begin{displaymath}
  \Sigma=\bigcup_{h=1}^{k} \{l_{i_h},l_{o_h}\}\cup O_h
  \qquad \text{where symbols $l_{i_h},l_{o_h}$ represent basic input and output activities}
\end{displaymath}
\[
  \text{consider the regular expressions on $\Sigma$}\qquad\qquad
E_h=(l_{i_h}\langle s. e_h l_{o_h}\rangle)^* \qquad e_h=(\sum_{o\in O_h}o.e_o)^*
\]
Name $s$ is a fresh session identifier which the symbol $\langle$
allocates when the session starts and the symbol $\rangle$ deallocates
when the session ends. Note that $s$ may occur in $e_h$ to e.g., avoid
re-authentication. Resources (activities, sessions, operations, etc.)
can be abstracted as letters and names, and the (de)allocation
represent the binding and freshness conditions. Intuitively, we give
the following nominal regular expressions with
$\Sigma=\bigcup_{h=1}^{k} \{l_{i_h},l_{o_h}\}\cup
\{\text{readFead},\text{updateProfile}\}$ and $n_1 \neq n_2$ be
distinct names. Operations $\text{readFeed}$, $\text{updateProfile}$
allow users read a feed and to update a profile.
\[
    E_1=(l_{i_1}\langle n_1. e_1 l_{o_1}\rangle)^* 
    \qquad
    e_1=(n\ \text{readFeed} + \text{updateProfile}\  E_2)^*
    \qquad
    E_2=(l_{i_2}\langle n_2. e_2 l_{o_2}\rangle)^*
\]
From the above equations, we see clearly the binders delimit the scope
of $n_1$ and $n_2$. And, $n_2$ is nested in $n_1$.
Intuitively, the approach above relaxes the condition of classical
language theory that the alphabet of a language is constant.
Binder allow us to extend the alphabet ``dynamically''.
This is strongly related to other approaches in the literature, where
languages over infinite alphabets are considered.
A form of regular expressions, called UB-expressions, for languages on
infinite alphabets investigated in~\cite{KaminskiT06}.
In~\cite{segoufin2006} pebble automata are compared to register
automata.
This class of languages are not suitable for our purposes as they do
not account for freshness.

Alternative approaches investigating languages over infinite alphabets
are those in~\cite{BojanczykKL14,SchroderKMW17}.
A finite representation of nominal sets and automata with data
symmetries and permutations is given in~\cite{BojanczykKL14};
this presentation differs from the one in~\cite{Kurz0T2012}.
Using the equivalence in~\cite{GadducciMM06}, the nominal regular
expressions with binders of~\cite{Kurz0T2012} can be transferred into
the context as the style of~\cite{BojanczykKL14}. In fact,
permutations permit to encode name binding and give an implicit
representation of name scoping in nominal words. On the other hand,
\cite{SchroderKMW17} defines \emph{bar strings} and considers another
representation of nominal sets and automata with binders, regular
expressions, and non-deterministic nominal automata over them.
An example of bar string with names $a$, $b$, and $c$ is $ab|ccb$
that represents a word where name $c$ following $a$ is bound in the
rest of the string.
A key observation is that in~\cite{BojanczykKL14} the scope of binders
models load freshness and it is fixed: once stated, the scope extends as far as possible ``to the
right''. In our setting, this would account to allocate a resource and never deallocate, that is, $|b$ in~\cite{BojanczykKL14} corresponds to the
notation $\ll b. b$ of~\cite{Kurz0T2012} where the angled bracket
opens the scope of the binder restricting the occurrences of name $b$
after the dot symbol; in this notation $\gg$ are used to close the
scope opened by $\ll$. To illustrate the differences
from~\cite{Kurz0T2012}, we consider the example under local freshness
semantics \cite[p. 5]{BojanczykKL14}. Let $\mathbb{A}$ be a set of
names, and $a$, $b$, and $c$ be names in $\mathbb{A}$, $\{|a|b,|a|a\}$
is alpha-equivalent to $\{|a|a\}$, unlike in~\cite{Kurz0T2012}
where is $\{\ll a.a\gg\ll b.b\gg,\ll a.a\gg\ll a.a\gg\}$ is
alpha-equivalent to $\{\ll a.a\gg\ll b.b\gg\}$ while
$\{\ll a.a\ll b.b\gg\gg,\ll a.a\ll a.a\gg\gg\}$ is alpha-equivalent to
$\{\ll a.a\ll b.b\gg\gg\}$. Thus, without the closing scope, the two
kinds of context are not easily transferred to each other.
Bar string $|a|ba$ could be corresponding to $\ll a.a \ll b.b\gg a\gg$
or $\ll a.a \ll b.ba\gg \gg$ under open conditions. However, if
$|a|ba$ is a word in the language
$\{cdc\in \mathbb{A}^3\mid c\neq d\}$, $|a|ba$ is corresponding to
$\ll a.a \ll b.b\gg a\gg$.
One of the most known and used learning algorithm is $\lstar$ introduced by
Angluin~\cite{Angluin87}.
As surveyed in \cref{sec:bkg}, given a regular language, $\lstar$
creates a deterministic automaton that accepts
the language.
This is done by mimicking the \quo{dialogue} between a \emph{learner}
and a \emph{teacher}; the former poses questions about the language to
the latter. In \lstar there are two types of queries the learner can
ask the teacher: \emph{membership} queries allow the learner to check
whether a word belongs to the input language while with
\emph{equivalence} queries the learner checks if an automaton accepts
or not the language.
The automaton is \quo{guessed} by the learner according to the answers
the teacher provides to queries.
The outcome of an equivalence query may be a \emph{counterexample}
selected by the teacher to exhibit that the automaton does not accept
the language.

The $\lstar$ algorithm has been extended to several classes of
languages~\cite{LCAfromMSCs,Niese2003,PasareanuGBCB08}.
Using a categorical approach, the $\lstar$ algorithm has been
generalised to other classes of automata such as Moore and
Mealy~\cite{LAinCP}.
An interesting line of research is the one explored
in~\cite{LCAfromMSCs} which applies learning automata to distributed
systems based on message-passing communications to learn communicating
finite-state machines~\cite{bz83} from message-sequence charts.
Applications of learning automata are in~\cite{PasareanuGBCB08}
and~\cite{Niese2003}.
The former defines a framework based on $\lstar$ to fully automatise
an incremental assume-guarantee verification technique and the latter
proposes an optimised approach for integrated testing of complex
systems.

% \begin{old}
% Recently,~\cite{Moerman17} proposes a learning algorithm which extends
% Angluin's $\lstar$ to nominal automata and~\cite{KlinS16} shows a
% non-trivial application to computation over infinite structures.
% %
% Language theories for infinite alphabets~\cite{BojanczykKL14,Pitts15}
% are used in~\cite{Moerman17}; names are handled through
% finitely-supported permutations which guarantee a finite
% representation of an infinite alphabet.
% %
% The automata over finite alphabets are represented by features of
% register automata~\cite{AartsFKV15,BolligHLM13,DAntoniV14}.
% %
% The states in nominal automata are as inducting the action of the
% group of permutations.

% <<<<<<< HEAD
% ~\cite{BolligHLM13} is a work to derive the Angluin’s algorithm for languages over infinite alphabets. Their target is one kind of register automata, the session automaton supporting the notion of fresh data values. The session automata are also over the alphabet with two sets, one is finite and another is infinite. The words are made by cartesian product of the two sets. Therefore, the corresponding data structure and definitions are changed. And,~\cite{BolligHLM13} reconstructs the Angluin’s algorithm focusing on the closeness of the observation tables. Interestingly, they infer the session automata in terms of canonical session automata, in order to require an equivalence query. Same as our work, the counterexamples are used to extend the observation tables and the alphabet. But, the operation on counterexamples is different in detail. 
% \end{old}

  Variants of Angluin’s algorithm for languages
  over infinite alphabets have attracted researchers'
  attention.
  An $\lstar$ algoritm for register automata is given
  in~\cite{BolligHLM13} where so-called \emph{session automata}
  support the notion of fresh data values. Session automata are
  defined over pairs of finite-infinite alphabets.
  %
  % Words are made
  % by cartesian product of two sets. Therefore, the corresponding
  % data structure and definitions are changed. And,~\cite{BolligHLM13}
  % reconstructs the Angluin’s algorithm with the redefined conditions
  % of the termination, focusing on the closeness of the observation
  % tables.
  %
  Interestingly, session automata also have a canonical form to decide
  equivalence queries.

  Like~\cite{BolligHLM13},~\cite{CasselHJS16} works on register
  automata and data language but the latter aims to
  the application of dynamic black-box analysis.
  The key point is that~\cite{CasselHJS16} uses a tree queries instead
  of membership queries and ensures the observation tables closeness
  and register-consistency. Further,~\cite{CasselHJS16} defines a new
  version of equivalence to achieve the correctness and termination of
  the algorithm.

  Recently,~\cite{Moerman17} proposes a learning algorithm which
  extends $\lstar$ to nominal automata. The main difference between
  our approach and the one in~\cite{Moerman17} is the representation
  of nominal languages and nominal automata. Language theories for
  infinite alphabets~\cite{BojanczykKL14,Pitts15} are used
  in~\cite{Moerman17}, handling names through finitely-supported
  permutations. Accordingly, in~\cite{Moerman17} observation tables
  and the states of nominal automata are orbit-finite. Another
  difference is the operation on counterexamples. Unlike in
  $\lstar$,~\cite{Moerman17} adds the counterexamples into columns. It
  is an interesting research direction to optimise our work on the
  operations of the counterexamples in the future.

\paragraph{Main contributions}
Our main objective is to develop a learning algorithm for nominal
languages, with a focus on resources.
The baseline to tackle this
objective is the use of binders and languages on infinite
alphabets.
\cref{sec:nla} collects our main results.

Our first achievement is the design of a learning algorithm that
generalises Angluin's \lstar\ algorithm to nominal regular languages
with binders.
We call our algorithm $\nlstar$ (after \emph{nominal} \lstar), as a
tribute to Angluin's work.
This is attained by retaining the basic scheme of $\lstar$ (query /
response dialogue been a learner and a teacher) and ideas of $\lstar$
(the represenation of a finite state automaton with specific a
\emph{observation table}).
Technically, this requires a revision of the main concepts of Angluin's
theory.
In particular, the type of queries and answers now have to account for
names and the allocation and deallocation operations on them.
Consequently, we have to reconsider the data structure to represent
observation tables and hence the notions of \emph{closedness} and
\emph{consistency}.

Interestingly, this revision culminates in
\cref{thm:nTabledeltastar}
and  \cref{thm:correct} respectively showing how the learned nominal automata
associated behave and  the correctness of \nlstar.
Finally, we discuss the complexity of \nlstar.

%%% Local Variables:
%%% mode: latex
%%% TeX-master: "main"
%%% End:

%% file: bkg.tex
We survey the principal concepts need in the rest of the paper.  In
particular, we review basics of formal language theory, its nominal
counterpart, and the learning algorithm of Angluin's \lstar~\cite{Angluin87}.

\subsection{Regular Languages}
Regular Languages, regular expressions and finite automata have a
well-known relationship established by the Kleene
theorem~\cite{kle56}.
A regular language can be represented by regular expressions and
accepted by a finite state automaton. In this section, we introduce
necessary notions and definitions for these concepts.

An \emph{alphabet} is a set (whose elements are often called
\emph{letters} or \emph{symbols}).
We denote a finite alphabet as $\Sigma$.
A \emph{word} is a sequence of symbols of an alphabet. Let $w$ be a
word, we denote the length of $w$ as $\length{w}$. The word of
length zero is called \emph{empty word} and denoted by
$\emptyword$. The concatenation of two words is denoted as
$\_\cdot\_$.
We define $\Sigma^\star=\bigcup_{n=0}^{\infty}\Sigma^n$, where
$\Sigma^0=\{\emptyword\}$ and for each
$n>0,\ \Sigma^n=\{w \cdot w' \sst w\in \Sigma\ and\ w'\in
\Sigma^{n-1}\}$. A \emph{language} $L$ is a set of words over an
alphabet $\Sigma$, that is, $L\subseteq\Sigma^\star$.

Other language operations we use are concatenation, union,
\emph{Kleene-star} and complementation. Assuming that $L\ and\ L'$ are
languages over $\Sigma$, we have the following standard definitions:
\begin{itemize}
\item concatenation $L\cdot L'=\{w\cdot w' \sst w\in L \ and \ w'\in L'\}$,
\item union $L\cup L'=\{w \sst w\in L \ or \ w\in L'\}$,
\item \emph{Kleene-star} $L^\star=\bigcup_{n=0}^{\infty}L^n=\left \{
    \begin{tabular}{cc}
      $\{\emptyword \}$ & $n=0$  \\
      $L^{n-1}\cdot L $ & $n\neq 0$  
    \end{tabular}
  \right .$,
\item complementation $L^C=\{w\in \Sigma^\star \sst w\notin L\}$.
\end{itemize}
A \emph{regular expressions over $\Sigma$} is a term derived from the
grammar where $a\in \Sigma$:
\[
  re \grmeq \epsilon \grmor \emptyset \grmor a \grmor  re+re \grmor  re \cdot re \grmor re^\star
\]
(where operators are listed in inverse order of precedence).
Given two regular expressions $re$ and $re'$, $re+re'$ denotes the
union of $re$ and $re'$, $re\cdot re'$ denotes the concatenation of
$re$ and $re'$, and $re^\star$ denotes the Kleene-star of $re$.

A \emph{finite automaton over alphabet $\Sigma$} is a five-tuple
$M=\langle Q,q_0,F,\delta\rangle$ such that $Q$ is a finite set of
states, $q_0$ is the initial state, $F\subseteq Q$ is the finite set
of final states, $\delta \subseteq Q\times \Sigma\times Q$ is a
relation from states and alphabet symbols to states.
The automaton $M$ is \emph{deterministic} when $\delta$ is a function
on $Q\times\Sigma\rightarrow Q$.
We extend the transition relation $\delta$ to $\Sigma^\star$ in the
obvious way, define the language of an automaton $M$ as usual, and
denote it as $\lofm{M}$.

It is well-known that regular expressions denote regular languages; we
let $\languageof{re}$ denote the language of a regular expression
$re$.
\begin{theorem}[\cite{Hopcroft1979}]\label{thm:minimal_unique}
  A language $L$ is regular iff there exists a finite automaton $M$
  such that $L=\lofm{M}$.
  Moreover, there exists a minimal finite automaton $M$ accepting $L$
  and $M$ is unique.
\end{theorem}

\subsection{Nominal Languages}\label{st:NL}
We use the nominal regular expressions introduced
in~\cite{Kurz0T2012,Kurz0T13}.
Languages over infinite alphabets are generalised to nominal automata
and nominal expressions~\cite{BojanczykKL14,Kurz0T2012,Kurz0T13,SchroderKMW17,segoufin2006}.
The approach in~\cite{Kurz0T2012} is distinguished by the use of names
and binders in the expressions.
In the following, we recall the basic notions first and then we survey
nominal languages~\cite{Kurz0T2012}.
Hereafter, we fix a countably infinite set of names $\Nom$.

A \emph{nominal language over $\Nom$ and $\Sigma$} is a set of nominal
words $w$ over $\Nom$ and $\Sigma$, that is terms derived by the
grammar
\[
  w \grmeq \epsilon \grmor a \grmor \aname \grmor w\cdot w \grmor \ll \aname.w\gg
  \qquad
  \text{where } \aname \in \Nom \text{ and } a \in \Sigma
\]
A name is \emph{bound} in a word when it occurs in the scope of a
binder.
Occurrences of names not bound are called \emph{free}.
For example, $\aname$ is bound in word
$\ll \aname. \aname \, a \gg \, \aname[m]$ while $\aname[m]$ is free.

A \emph{nominal regular expression} is a term derived from the grammar
\[
  ne \grmeq \epsilon \grmor \emptyset \grmor a \grmor \aname \grmor ne+ne  \grmor ne \cdot ne \grmor ne^\star \grmor \langle \aname.ne \rangle
  \qquad
  \text{where } \aname \in \Nom \text{ and } a \in \Sigma
\]
In nominal expressions, binders are represented as
$\langle \aname.\_ \rangle$ for $\aname \in \Nom$.
If the names in a nominal expression are all bound, the nominal
expression is \emph{closed}.
Nominal regular expressions denote \emph{nominal languages}.
\begin{definition}\cite{Kurz0T2012}\label{df:neTolanguage}
  The \emph{nominal language $\languageof{ne}$ of a nominal regular
    expression $ne$} is defined as
  \begin{itemize}
  \item $ \languageof{\emptyword}=\{\emptyword\} \qquad \languageof{\emptyset}=\emptyset \qquad \languageof{a}=\{a\} \qquad \languageof{\aname}=\{\aname\}$
  \item $\languageof{ne_1+ne_2}=\languageof{ne_1}\cup \languageof{ne_2}$
  \item $\languageof{ne_1\cdot ne_2}=\languageof{ne_1}\cdot \languageof{ne_2}=\{w\cdot v \sst w\in \languageof{ne_1},v\in \languageof{ne_2}\}$
  \item $\languageof{ne^\star}=\bigcup\limits_{k\in \Nat} \languageof{ne}^k$, where $\languageof{ne}^k= 
      \left \{
        \begin{tabular}{cc}
          $\{\epsilon\}$ & $k=0$  \\
          $\languageof{ne}\cdot \languageof{ne}^{k-1} $ & $k\neq 0$  
        \end{tabular}
      \right .
      $
    \item $\languageof{\langle \aname . ne \rangle}=\{\ll \aname.w\gg \sst w\in \languageof{ne}\}$.
  \end{itemize}
\end{definition}
The closure properties of nomianl regular languages are stated below:
\begin{theorem}~\cite{Kurz0T2012}\label{thm:closureprops}
  Nominal regular languages are closed under union, intersection, and \emph{resource sensitive} complementation. 
\end{theorem}
The main difference with respect to classical regular expressions is
on complementation.
Since, the complement of $\languageof{ne}$ is not a nominal regular
language, the classical complementation does not work in the nominal
case.
Therefore,~\cite{Kurz0T2012} give the following definition:
\begin{definition}~\cite{Kurz0T2012}\label{def:rsComplementation} Let
  $ne$ be a nominal regular expression.
  The \emph{resource sensitive complement of $\languageof{ne}$} is the
  set $\{w \notin \languageof{ne} \sst \theta(w) \leq \theta(ne)\}$
  where
  \begin{itemize}
  \item
    $ne \in \{\emptyword,\emptyset\} \cup \Nom \cup \Sigma \implies
    \theta(ne) =0$
  \item $ne=ne_1+ne_2$ or
    $ne_1\cdot ne_2 \implies \theta(ne) =max(\theta(ne_1) , \theta(ne_2) )$
  \item $ne = \langle \aname. ne \rangle \implies 1 + \theta(ne)$
  \item $ne= ne^\star \implies \theta(ne) $.
  \end{itemize}
  and the depth $\theta$ of a word is defined as the depth the corresponding
  expression.
\end{definition}

We now define the notions of nominal automata adopted here. Let
$\Nat$ be the set of natural numbers and define
$\underline{n}=\{1,\cdots,n\}$ for each $n\in \Nat$. Considering
a set of states $Q$ paired with a map $\|\_\|: Q\rightarrow \Nat$,
let us define the local registers of $q\in Q$ to be
$\underline{\|q\|}$. We use a definition of nominal
automata~\cite{Kurz0T12} as Definition~\ref{def:NA}. Moreover, we
describe how to allocate names via maps
$\sigma:\underline{\|q\|}\rightarrow \Nom$.

\begin{definition}[Nominal Automata~\cite{Kurz0T12}]\label{def:NA}
  Let $\Nf \subset \Nom$ be a finite set of names. A
  \emph{nominal automaton with binders over $\Sigma$ and $\Nf$},
  $(\Sigma, \Nf)$-automaton for short, is a tuple
  $M=\langle Q,q_0,F,\delta\rangle$ such that
  \begin{itemize}
  \item $Q$ is a finite set of states equipped with a map $\|\_\|: Q\rightarrow \Nat$
  \item $q_0$ is the initial state and $\|q_0\|=0$
  \item $F$ is the finite set of final states and $\|q\|=0$ for each $q\in F$
  \item for each $q\in Q\ and\ \alpha\in \Sigma\cup\Nf \cup{\{\emptyword,\ll,\gg\}}$, we have a set $\delta(q, \alpha)\subseteq Q$ such that for all $q'\in \delta(q, \alpha)$ must hold:
      \begin{itemize}
      \item $ \alpha =\ll$ $\implies \ \|q'\|=\|q\|+1$
      \item $ \alpha=\gg$ $\implies \ \|q'\|=\|q\|-1$
      \item otherwise $\implies \ \|q'\|=\|q\|$
      \end{itemize}
      A transition is a triple $(q, \alpha,q')$ such that $q'\in \delta(q, \alpha)$.
  \end{itemize}
\end{definition}
A nominal automaton $M$ is \emph{deterministic} if, for each $q\in Q$,
\[
  \left\{ \begin{array}{ll}
            \textrm{$| \delta(q, \alpha)| =0$,} & \textrm{if ($ \alpha =\ll$ and $\|q\|=max\{\|q'\|\ |\ q'\in Q\}$) or ($ \alpha =\gg$ and $\|q\|=0$)}\\
            \textrm{$| \delta(q, \alpha)| =1$,} & \textrm{otherwise}
          \end{array} \right..
\]
Let $M=\langle Q,q_0,F,\delta\rangle$ be a nominal automata over
$\Sigma$ and $\Nf$, we denote the image of a map $\sigma$ by
$Im(\sigma)$ and the empty map by $\emptyset$. Let $q$ be a state, $w$
be a word whose free names are in $\Nf\cup Im(\sigma)$ and
$\sigma:\underline{\|q\|}\rightarrow \Nom$ be a map, a
configuration of $M$ is denoted by $\langle q,w,\sigma\rangle$. A
configuration $\langle q,w,\sigma\rangle$ is \textit{initial} if
$q=q_0$, $w$ is a word whose free names are in $\Nf$, and
$\sigma=\emptyset$; a configuration $\langle q,w,\sigma\rangle$ is
accepting if $q\in F$, $w=\emptyword$, and $\sigma=\emptyset$.
%\begin{definition}\label{def:AllocatingNames}
  Given $q,q'\in Q$ and two configurations
  $t=\langle q,w,\sigma\rangle$ and
  $t'=\langle q',w',\sigma' \rangle$, $M$ \emph{moves from $t$ to
    $t'$} if there is
  $s \in \Sigma \cup \Nom \cup \{\emptyword,\ll,\gg\} \cup \Nat$
  such that $q'\in \delta(q,s)$ and
  \[\begin{cases}
      s \in \underline{\|q\|}, & w=\sigma(s)w',\ \sigma'=\sigma \text{ and } \forall n > s \qst \sigma(s) \neq \sigma(n)
      \\
      s \in \Nf \setminus Im(\sigma) & w=aw',\ \sigma'=\sigma
      \\
      s \in \Sigma & w=aw',\ \sigma'=\sigma
      \\
      s = \emptyword & w=w',\ \sigma'=\sigma
      \\
      s = \ll & w=\ll w',\ \sigma' = \sigma \lbrack \|q'\| \mapsto n \rbrack
      \\
      s = \gg & w=\gg w', \ \sigma' = \sigma_{|_{\underline{\|q'\|}}}
    \end{cases}
  \]
  where $\sigma\lbrack\|q'\| \mapsto n\rbrack$ extends $\sigma$ by
  allocating the maximum index in $\underline{\|q\|}$ to $n$ and
  $\sigma_{|_{\underline{\|q'\|}'}}$ is restriction on
  $\underline{\|q'\|} $ of $\sigma$.
  The language accepted by $M$ is the set of nominal words $w$ such
  that $M$ moves from the initial configuration
  $\langle q,w,\sigma\rangle$ to an accepting configuration.
  (For more details see~\cite{Kurz0T12,Kurz0T13}).
  
  % \end{definition}

%  Theorem~\ref{thm:kleeneforNA} transfers to nominal regular expressions  the classical theorem of Kleene
  \begin{theorem}~\cite{Kurz0T2012}\label{thm:kleeneforNA} Every
    language recognised by a nominal automaton is representable by a
    nominal regular expression. Conversely, every language represented
    by a nominal regular expression is acceptable by a nominal
    automaton.
\end{theorem}

\subsection{Angluin's Algorithm $\lstar$}\label{st:Lstar}
The algorithm $\lstar$ was introduced in~\cite{Angluin87} to learn a
finite automaton accepting a given regular language $L$ over an
alphabet $\Sigma$.
The basic idea of the algorithm is to implement a dialogue between a
\emph{learner} and a \emph{teacher}.
The learner may ask the teacher for \textit{membership} queries
``$w \in L$?'' to check whether a word $w$ is in the given
language.
Moreover, the learner may submit an automaton $M$ to the teacher who
replies ``yes'' if $\lofm{M} = L$, or provides a counter-example
showing that $\lofm{M} \neq L$.
The teacher is assumed to answer all the learner's questions
correctly.

Key data structures of $\lstar$ are \emph{observation tables}
representing finite predicates of words over $\Sigma$ classifying them
as members of $L$ or not.
\begin{definition}[Observation Tables~\cite{Angluin87}]
  An \emph{observation table} $(S,E,T)$ consists of nonempty finite languages $S, E\subseteq\Sigma^\star$ such that $S$ is prefix-closed and $E$ is suffix-closed, and a function $T:(S\cup S\cdot \Sigma)\cdot E\rightarrow \{0,1\}$.
\end{definition}
The rows of an observation table are labelled by elements of
$S\cup S\cdot \Sigma$, and the columns are labelled by elements of $E$
with the entry for row $s$ and column $e$ given by $T(s\cdot e)$.
A row of the table can be represented by a function
$\row(s): E\rightarrow \{0,1\}$ such that $\row(s)(e)=T(s\cdot e)$.
A word $s\cdot e$ is a member of $L$ of $(S,E,T)$ iff
$T(s\cdot e)= 1$.  An observation table $(S,E,T)$ is \emph{closed}
when
\[
  \forall w\in S\cdot \Sigma\,.\,\exists s\in S\,.\,\row(w)=\row(s)
\]
An observation table $(S,E,T)$ is \emph{consistent} when for all
$a \in \Sigma$ and all $s,s'\in S$
\[
  \row(s)=\row(s')\implies \row(sa)=\row(s'a)
\]
A closed and consistent observation table $(S,E,T)$ has an associated
finite automaton $M=(Q,\delta,q_0,F)$ given by
\begin{itemize}
\item $Q =\{\row(s) \sst s \in S\}$,
\item $q_0 = \row(\emptyword)$,
\item $F =\{\row(s) \sst \row(s)(\emptyword)=1, s \in S \}$,
\item {$\delta(\row(s),a)=\row(s\cdot a),\ a\in \Sigma$.}
\end{itemize}

To see that this is a well-defined automaton, note that the initial state
is defined since $S$ is prefix-closed and must contain
$\emptyword$. Similarly, $E$ is suffix-closed and must contain
$\emptyword$. And, if $s,s'\in S, \row(s)=\row(s')$, then
$T(s)=T(s\cdot\emptyword)$ and $T(s')=T(s'\cdot\emptyword)$ are equal
as defined.
The transition function is well-defined since the table is closed and
consistent. Suppose $s$ and $s'$ are elements of $S$ such that
$\row(s)=\row(s')$. Since the table $(S,E,T)$ is consistent,
$\forall a\in \Sigma, \row(sa)=\row(s'a) $. And the value of
$\row(sa)$ is equal to such a $\row(s'')$ for an $s''\in S$, since the
table is closed.

\begin{figure}[h]
  \centering
  \begin{algorithmic}[1]
    
    % \Statex{Input: a regular language $L$ on the alphabet $\Sigma$.}
    % \Statex{Output: an finite automaton accepting $L$.}
    \State{Initialisation: $S=\{\emptyword\}, E=\{\emptyword\}$.}
    \State{Construct the initial observation table $(S,E,T)$ by asking for membership queries about $(S\cup S\cdot \Sigma)\cdot E$.}\label{ag:init}
    \Repeat\label{ag:mainstart}
    \While{$(S,E,T)$ is not closed or not consistent}
    \If{$(S,E,T)$ is not closed}\label{ag:or_closed}
    \State{find $s'\in S\cdot \Sigma$ such that\;}
    \State{$\row(s)\neq \row(s')$ for all $s\in S$ ,\;}
    \State{add $s'$ into $S$,\;}
    \State{extend $T$ to $(S\cup S\cdot \Sigma)\cdot E$ using membership queries.}
    \EndIf
    \If{$(S,E,T)$ is not consistent}\label{ag:or_consistent}
    \State{find $s_1,s_2\in S$,$e\in E$ and $a\in \Sigma$ such that\;}
    \State{$\row(s_1)=\row(s_2)$ and $\row(s_1\cdot a)(e)\neq \row(s_2\cdot a)(e)$,} 	
    \State{Add $a\cdot e$ into $E$,\;}
    
    \State{extend $T$ to $(S\cup S\cdot \Sigma)\cdot E$ using membership queries.}
    \EndIf
    \EndWhile
    \State{Construct an automaton $M$ from table $(S,E,T)$ and ask teacher an equivalence query.}
    \If{teacher replies a counterexample $c$}
    \State{add $c$ and all its prefixes into $S$.}
    \State{extend $T$ to $(S\cup S\cdot \Sigma)\cdot E$ using membership queries.}
    \EndIf
    \Until{teacher replies yes to equivalence query $M$.}\label{ag:mainend}
    \State{Halt and output $M$.}
  \end{algorithmic}
  \caption{The learner in $\lstar$.}\label{Angluin_LA}
\end{figure}

The learning process of \lstar is shown in Figure~\ref{Angluin_LA}.
Let  $L$ be the input regular language over an alphabet $\Sigma$.
Initially, the observation table $(\{\emptyword\},\{\emptyword\},T)$
is such that $T$ is initialised by asking for membership queries about
$\emptyword$ and each element in $\Sigma$ (line~\ref{ag:init}).
Then the algorithm enters into the main loop
(lines~\ref{ag:mainstart}-\ref{ag:mainend}). Inside of the main loop,
a while loop tests the current observation table $(S,E,T)$ for
closedness (line~\ref{ag:or_closed}) and consistency
(line~\ref{ag:or_consistent}).

If the current observation table $(S,E,T)$ is not closed, the
algorithm finds $s'$ in $S\cdot \Sigma$ such that $\row(s')$ is
different from $\row(s)$ for all $s\in S$. Then the word $s'$ is added
into $S$ and new rows are added for words $s'\cdot a$ for all
$a\in \Sigma$. Thus, $T$ is extended to $(S\cup S\cdot \Sigma)\cdot E$
by asking for membership queries about missing elements.

Similarly, if $(S,E,T)$ is not consistent, the algorithm finds $s_1,s_2\in S$,$e\in E$, and $a\in \Sigma$ such that $\row(s_1)=\row(s_2)$ but $\row(s_1\cdot a)(e)\neq \row(s_2\cdot a)(e)$. The word $a\cdot e$ is added into $E$. That is, each row in the table has a new column $a\cdot e$. $T$ is extended to $(S\cup S\cdot \Sigma)\cdot E$ by asking for missing elements $\row(s)(a\cdot e)$ for all $s\in (S\cup S\cdot \Sigma)$.

An associated automaton $M$ is constructed when the observation table
$(S,E,T)$ is closed and consistent. And then, an equivalence query
about $M$ is asked for. The algorithm terminates and outputs $M$ when
the teacher replies ``yes" to the query. If the teacher replies with a
counterexample $c$, the word $c$ and all its prefixes are added into
$S$, and then $T$ is extended by asking membership queries about new
entries in $(S\cup S\cdot \Sigma)\cdot E$. Then, a new round for the main
loop of closedness and consistency starts.
% We give an example of how the algorithm works in \cref{sec:AngluinAlgorithm}.
% Angluin orginal algorithm~\cite{Angluin87} think about stochastic setting to finding counterexamples. The Latest~\cite{Angluin17} discuss proper equivalence queries with random counterexamples. Instead, our research dose not choose random counterexamples and we will introduce it in Section~\ref{st:findingC}. Assuming that $M$ is learner's automaton and $L$ is the target language, the approach in~\cite{Angluin87} is roughly finding a word $w$ in some probability distribution on the set of all words over $\Sigma$, which $w$ is accepted by either $M$ or $L$.

%%% Local Variables:
%%% mode: latex
%%% TeX-master: "main"
%%% End:

%% file: nla.tex
In this section, we introduce our learning algorithm based on nominal
automata. Our teacher still answers two kinds of queries: membership
queries and equivalence queries regarding a target nominal regular
language.

\input{supp}

Fix a nominal regular language $L$, in our algorithm, the learner asks for membership queries about legal
words. If a word $w$ is not legal, the learner marks it as
$\bot$ in the observation table. The membership query consists of a
legal word $w$ and it has the following possible answers:
\begin{itemize}
\item if $w\in L$, the answer is \quo{1},
\item if $w$ is a prefix of a word in $L$, the answer is \quo{$P$},
\item otherwise, the answer is \quo{0}.
\end{itemize}

As in Angluin's \lstar algorithm, only the teacher knows $L$.
Unlike in \lstar, the learner in our algorithm does not know the whole
alphabet $A_n$.
The learner knows $\Sigma$ initially and learns names via
counterexamples.
We will see that the learner knows the whole alphabet
when the algorithm terminates.
             % research questions here?
             % give a brief explanation about why we let learner learn the alphabet not know it.
             % give a brief how we reach the final version among many possibilities.
  \paragraph{Remark.}The answer \quo{P} is used for efficiency. In
  fact, the teacher could answer \quo{0} instead of \quo{P}. However,
  this would require the learner to ask more membership or equivalence
  queries.
  This is confirmed by some experimental results that are not in scope of this paper.
  % A comparison is provided in Table~\ref{tb:comparison}.
\subsection{Nominal observation tables}
Observation tables are pivotal data structure to ensure the
algorithm's functionalities. A closed and consistent observation table
allows us to construct a minimal automaton.
We extend Angluin's observation tables to \emph{nominal observation
  tables}, \nOtable[s] for short.
\begin{definition}[\nOtable]\label{def:nOT}
A \emph{legal word} is a prefix of a nominal word; the \emph{depth}
  $\|w\|$ of a legal word $w$ is the highest number of nested binders
  in $w$.
  Let
  \[
    \Alp_0 = \Sigma
    \qquad\text{and}\qquad
    \Alp_n= \{\ll,\gg\} \cup \Sigma\cup \underline{n}
    \quad \text{ for } 0 < n \in \Nat
  \]
  A tuple $\Otable$ is an
  \emph{\nOtable} if
  \begin{itemize}
  \item{$S\subseteq {\Alp_n}^*$ is a prefix-closed set of legal strings, for all $s\in S$,$\reg{s}\leq n$, }
  \item{$E\subseteq {\Alp_n}^*$ is suffix-closed,}
 % \item{$lv:\reg{s}\rightarrow \mathbb{N}$,}  
  \item{$T:(S\cup S\cdot {\Alp_n})\cdot E\rightarrow \labelsinT$.}
  \end{itemize}
\end{definition}
As in Angluin's definition, an \nOtable $\Otable$ consists of rows
labelled by legal words in $S\cup S\cdot \Alp_n$ and columns labelled
by words in $E$:
\begin{gather*}
  \row {}:{}   (S\cup S\cdot \Alp_n) \to (E \to \labelsinT)\\
  \row(s)(e)= T(s\cdot e)
\end{gather*}
In order to reflect the layers of nominal automata, we use $\reg{\_}$ to distinguish rows. Therefore, we need the following auxiliary notion of equivalence of rows:
in an \nOtable $\Otable$, for all $ s, s'\in S\cup S\cdot \Alp_n$, $$\row(s)\doteq\row(s')\ \Longleftrightarrow\ \row(s)=\row(s')\  and\ \reg{s}=\reg{s'}.$$ 
Accordingly, the definition of closed and consistent table changes as follows.
\begin{definition}[Closed and Consistent Tables]
  An \emph{\nOtable} $\Otable$ is \emph{closed} when
  $$\forall s'\in S\cdot {\Alp_n}.\exists s\in S.\ \ \row(s')\doteq\row(s).$$ 
  An \emph{\nOtable} $\Otable$ is \emph{consistent} when $$\forall \alpha\in {\Alp_n}.\forall s,s'\in S\ \ \row(s)\doteq\row(s')\implies \row(s\alpha)\doteq\row(s'\alpha). $$
\end{definition}
\subsection{From \nOtable[s] to nominal automata}
Analogously to Angluin's theory, closed and consistent \nOtable[s]
correspond to deterministic finite nominal automata.
\begin{definition}
The ($\Sigma,\mathcal{N}_{fin}$)-automaton $M=(Q,q_0,F,\delta)$ associated with a closed and consistent \emph{\nOtable} $\Otable$ is defined as 
\begin{itemize}
\item{$\Sigma={\Alp_n}\setminus\{\{\ll,\gg\}\cup\underline{n}\}$, $\mathcal{N}_{fin}=\underline{n}$,}
\item{a set of states $Q=\{(\row(s),\reg{s})\mid s\in S \}$ with a map $\Mreg{\_}$, and $\Mreg{q}=\reg{s}$ for each $q=(\row(s),\reg{s})\in Q$,}
\item{an initial state $q_0=(\row(\emptystring),\reg{\emptystring})$,}
\item{a set of final states $F=\{(\row(s),\reg{s})\mid \row(s)(\emptystring)=1,\reg{s}=0\ and\ s\in S \}$, }
%\item{for each $q\in Q$, we have a set $\delta(q,\alpha)\subseteq Q$ over ; for all $q'\in \delta(q,\alpha)$ must hold: $q=\state{s}, s\in S\implies s\cdot \alpha\in S\cup S\cdot A_n$. A transition is a triple $(q,\alpha,q')$ such that $q'\in \delta(q,\alpha)$.}
\item{A transition function is a partial function $\delta:Q\times {\Alp_n}\rightarrow Q$: for all $s\in S,\alpha\in {\Alp_n}$,
$\delta((\row(s),\reg{s}),\alpha)=
                   (\row(s\alpha),\reg{s\alpha})$ if $s\alpha\in S\cup S\cdot {\Alp_n}$.
}
\end{itemize}
\end{definition}

Accordingly, we define a partial function $\delta^\ast: Q\times {\Alp_n}^* \to Q$ inductively as follows
\begin{align*}
  \delta^\ast(q,\emptyword) &= q\\
  \delta^\ast(q,aw) &= \delta^\ast(\delta(q,a),w)
\end{align*}
for all $a\in \Alp_n,\ w\in {\Alp_n}^*$, $q\in
Q$.  Note that $\delta^*(q,a)=\delta^*(q,a\cdot
\emptyword)=\delta^*(\delta(q,a),
\emptyword)=\delta(q,a)$.
\begin{lemma}\label{lem:nTabledeltastar}
  Let $M=(Q,q_0,F,\delta)$ be the automaton associated with a closed and consistent \nOtable $\Otable$. Suppose $w,u\in{\Alp_n}^*$. We have  $\delta^\ast(q,w\cdot u)=\delta^\ast(\delta^\ast(q,w), u)$ for all $q\in Q$.
\end{lemma}

\begin{proof}
  By induction on length of $w$.
\end{proof}

\begin{theorem}\label{thm:nTabledeltastar}
  Assume that  $M=(Q,q_0,F,\delta)$ is the automaton associated with a closed and consistent \nOtable $\Otable$.
  \begin{itemize}
  \item
    For all $w$ in $S\cup S\cdot {\Alp_n}$, $\delta^\ast(q_0,w)=\state{w}$.
  \item
    For all $w$ in $S\cup S\cdot {\Alp_n}$ and $u$ in $E$, $\delta^\ast(q_0,w\cdot u)$ in $F$ if and only if $$\row(w)(u)=1.$$
  \end{itemize}
\end{theorem}
\begin{proof}
  Let $w=w'a$ in $S\cup S\cdot {\Alp_n}$ and $u=a\cdot u'$ in $E$.\\
  \\
  Since $S$ is prefix-closed, all prefixes of $w$ are in $S$, that is, $w'$ is in $S$. We know:
  \begin{align*}
    \delta^\ast(q_0,w)&=\delta^\ast(q_0,w'a)\\
                      &=\delta^\ast(\delta^\ast(q_0,w'),a)&& \textrm{by Lemma~\ref{lem:nTabledeltastar}}\\
                      &=\delta^\ast(\state{w'},a)&& \textrm{by induction hypothesis}\\
                      &=\delta(\state{w'},a) && \textrm{by the definition of $\delta^\ast$}\\
                      &=\state{w'a} && \textrm{by the definition of $\delta$}\\
                      &=\state{w}
  \end{align*}
  Since $E$ is suffix-closed, all suffixes of $u$ are in $E$. Depending on the length of $u$, we have two situations.
  \begin{itemize}
  \item When
    $u=\epsilon$,
    $\row(w)(u)=\row(w)(\emptyword)$ and
    $\delta^\ast(q_0,w\cdot
    u)=\delta^\ast(q_0,w)$. From preceding
    proof,
    $\delta^\ast(q_0,w)=\state{w}$. Because
    the table is closed, there is a $w'\in
    S$ such that $\row(w')\doteq
    \row(w)$. $\delta^\ast(q_0,w\cdot
    u)$ is in $F$ if and only if
    $\row(w')$ is in
    $F$ from the definition of
    $F$. Thus
    $\row(w')(\emptyword)=\row(w)(\emptyword)=1$,
    that is, $\row(w)(u)=1$.
  \item{Assume that when the length of $u'\in E$ is $n$, we have that for all $w$ in $S\cup S\cdot {\Alp_n}$ and $u$ in $E$, $\delta^\ast(q_0,w\cdot u)$ in $F$ if and only if $\row(w)(u)=1$. Let $u=au'$ and $u\in E$. Because the table is closed, there is a $w'\in S$ such that $\row(w')\doteq\row(w)$.
      \begin{align*}
        \delta^\ast(q_0,w\cdot u)
        &=\delta^\ast(\delta^\ast(q_0,w), u) && \textrm{by Lemma~\ref{lem:nTabledeltastar}}\\
        &=\delta^\ast(\state{w}, u) && \textrm{by preceding proof}\\
        &=\delta^\ast(\state{w'},u) && \textrm{since $\row(w')=\row(w)$}\\
        &=\delta^\ast(\state{w'},au') && \textrm{$u=au'$}\\
        &=\delta^\ast(\row(w'\cdot a),u') && \textrm{by closedness and definition of $\delta$}\\
        &=\delta^\ast(\delta^\ast(q_0,w'\cdot a),u')&& \textrm{by preceding proof}\\
        &=\delta^\ast(q_0,w'\cdot a\cdot u')
      \end{align*}
    }
  \end{itemize}
  By induction hypothesis on $u'$, $\delta^\ast(q_0,w'\cdot a\cdot
  u')$ is in $F$ if only if $\row(w'\cdot a)(
  u')=1$. Because $\row(w)\doteq \row(w')$ and $u=au'$, $\row(w'\cdot
  a)( u')=T(w'\cdot a\cdot u')=\row(w')( a\cdot
  u')=\row(w')(u)=\row(w)(u)$. Therefore $\delta^\ast(q_0,w\cdot
  u)$ in $F$ if and only if $\row(w)(u)=1$.
\end{proof}

We are now ready to introduce a learning algorithm for our nominal
automata.

\section{The \nlstar\ Algorithm}
We dubbed our algorithm \nlstar, after \emph{nominal} \lstar.
The algorithm is shown in Figure~\ref{LANA_a3}.
The learner in \nlstar is similar to the one in \lstar.
Basically, our learner modifies the initial \nOtable until it becomes
closed and consistent in the nominal sense (according to notions
introduced before).
When the current \nOtable $\Otable$ is closed and consistent, the
learner would ask the teacher if the automaton associated with $\Otable$
accepts the input language $L$.
If this is the case, the teacher will reply \quo{yes} and the learning process halts.
Otherwise, the learning process continues after the teacher has
produced a counterexample.

Because of the new definitions of closedness and consistency, we
refine some actions about checking closedness
(line~\ref{codeline:closedchanges}) and consistency
(line~\ref{codeline:findrowsconsistency}).
If the current \nOtable $\Otable$ is not closed, the learner finds a
row $s'$ such that for no $s\in S$ we have $\row(s') \doteq \row(s)$.
If $\Otable$ is not consistent, the learner finds a word $a\cdot e$,
with $a\in \Alp_n$ and $e\in E$, such that for some $s_1\in S$ and
$s_2\in S$ with $\row(s_1) \doteq \row(s_2)$, we have
$\row(s_1\cdot a)(e)\neq\row(s_2\cdot a)(e)$.

\begin{figure}
  \centering\small
  \begin{algorithmic}[1]
    \State{Initialisation: $S=\{\emptyword\}, E=\{\emptyword\}$, $n=0$}\label{codeline:init}
    \State{Asking for membership queries about $\emptyword$ and each $a \in \Alp_n$ build the initial observation table $\Otable$.}
    \Repeat
    \While{$\Otable$ is not closed or consistent}\label{codeline:while1}
    \If{$\Otable$ is not closed}\label{codeline:closedstatement}
    \State{\textbf{find} $s'\in S\cdot \Alp_n$ such that $\row(s)\doteq\row(s')$ is not satisfied for all $s\in S$  \;}\label{codeline:closedchanges}
    \State{\textbf{add} $s'$ into $S$\;}\label{codeline:addstate}
    % \State{\textbf{extend} $\Alp_n$ with $\reg{s}$ for all $s\in S$, $\reg{s}>0$.}\label{codeline:alphabetchanges}
    \State{\textbf{extend} $T$ to $(S\cup S\cdot \Alp_n)\cdot E$ using membership queries.}\label{codeline:closedextension}
    
    \EndIf
    \If{$\Otable$ is not consistent}\label{codeline:consiststatement}
    \State{\textbf{find} $s_1,s_2\in S$,$e\in E$ and $a\in \Alp_n$ such that \label{codeline:findrowsconsistency}
      $\row(s_1)\doteq\row(s_2)$ but $\row(s_1\cdot a)(e)\neq \row(s_2\cdot a)(e)$.}\label{codeline:consistentchanges
    }
    
    \State{\textbf{Add} $a\cdot e$ into $E$\;}\label{codeline:addE}
    %	\EndIf
    \State{\textbf{extend} $T$ to $(S\cup S\cdot \Alp_n)\cdot E$ using membership queries.}\label{codeline:consistentextension}
    \EndIf
    \EndWhile
    \State{Construct an automata $M$ associated to $\Otable$.}
    \State{Ask equivalence query about $M$.}
    \If{teacher replies a counterexample $c$}\label{codeline:replycounterexample}
    \State{add $c$ and all its prefixes into $S$.}\label{codeline:addcounterexample}
    %	\State{A set $temp$=GENERATENAMES(c)}
    %	\State{$S=S\cup temp$, also add all prefixes of $temp$'s elements into $S$ }
    %	\If{$c$ includes $\langle$}
    % \State{calculate $n$ the number of continuous occurrences of $\langle$ in $c$}
    % \State{Add $n$ into $A$}
    %	\EndIf
    \State{\textbf{extend} $\Alp_n$ with $\reg{s}$ for all $s\in S$, $\reg{s}>0$.}\label{codeline:alphabetchanges2}
    \State{\textbf{extend} $T$ to $(S\cup S\cdot \Alp_n)\cdot E$ using membership queries.}
    \EndIf
    \Until{teacher replies yes to $M$.}\label{codeline:replyyes}
    
    \State{Halt and output $M$.}
  \end{algorithmic}
  % \end{algorithm}
  \caption{The learner of Learning Algorithm for Nominal Automaton}\label{LANA_a3}	
\end{figure}
The main difference with respect to the algorithm of Angluin is that
the learner has partial knowledge of the alphabet.
The alphabet $\Alp_n$ is enlarged by adding names
(line~\ref{codeline:alphabetchanges2}) during the learning
process. More precisely, the learner expands the alphabet if the
counterexample requires to allocate fresh names.
When names are required, the operators of allocations and deallocations are added into $\Alp_n$. When the algorithm terminates, the learner's alphabet $\Alp_n$ is the
alphabet of the given language. Like in the original algorithm, our
algorithm terminates when the teacher replies \quo{yes} to an equivalence
query.

We now show that \nlstar\ is correct.
That is, that eventually the teacher replies \quo{yes} to an equivalence query.
In other word, \nlstar\ terminates with a \quo{yes} answer to an equivalence query.
Hence, the automaton submitted in the query accepts the input language.

\begin{theorem}\label{thm:correct}
  The algorithm terminates, hence it is correct.
\end{theorem}
\begin{proof}
  We show that the if- and the while-statements terminate.
  Let us consider the if-statements first.
  It is easy to check that  closedness and
  consistency are decidable because these properties require just the inspection of the
  \nOtable (which is finite).
  Hence, the if-statements starting at
  lines~\ref{codeline:closedstatement}
  and~\ref{codeline:consiststatement} never diverge because their
  guards do not diverge and their then-branch is a finite sequence of
  assignments:
  \begin{itemize}
  \item The if-statement for closedness
    (line~\ref{codeline:closedstatement}) terminates directly if the
    table is closed. Otherwise, in case of making a table closed, we
    find a row $s' \in S \cdot \Alp_n$ such that $\row(s')\doteq\row(s)$
    is not satisfied for all $s\in S$. The algorithm adds $s'$ into
    $S$. Since $\Alp_n$ is finite and bounded by $n$, the sets $S$ and
    $E$ are both finite. Thus, $S \cdot \Alp_n$ is a finite set and
    there are finitely many choices for $s'$. That is, 
    line~\ref{codeline:addstate} can only be executed finitely
    times. Besides, the content of rows is one of the permutations and
    combinations of $\labelsinT$ which also has finite
    possibilities. So we conclude that the branch terminates.
  \item Similarly, the if-statement for consistency
    (line~\ref{codeline:consiststatement}) terminates directly if the
    table is consistent. Otherwise, to make the table consistent the
    algorithm searches for two rows $s_1,s_2\in S$ satisfying the
    condition at line~\ref{codeline:findrowsconsistency}.
    As in the previous case for closedness, to add elements into $E$
    there are only finitely many possibilities $s_1$, $s_2$, $a$, and
    $e$ (line~\ref{codeline:addE}). Thus, the branch of the
    if-statement terminates.
\end{itemize}
  Therefore, the while-statement (line~\ref{codeline:while1}) terminates in finite repetitions, since the algorithm makes a table closed and consistent in finite operations. Then, the algorithm will succeed in construct an automata $M$ associated to a closed and consistent table. Next, the learner asks for an equivalence query. The teacher replies a counterexample $c$ (line~\ref{codeline:replycounterexample}) or yes (line~\ref{codeline:replyyes}). It remains to prove that the learner only asks finitely many equivalence queries.

  Let $M$ be the nominal automaton associated to the current \nOtable
  $\Otable$. Assume that the equivalence query about $M$ fails. The
  teacher has to find a counterexample $c$; this is finitely
  computable since the nominal regular expressions are closed under
  the operations of Kleene algebra and under resource complementation.
  Hence, the if-statement on line~\ref{codeline:replycounterexample}
  goes the branch extending the table
  (line~\ref{codeline:addcounterexample}). Then, the algorithm will
  start a new loop (line~\ref{codeline:while1}) for the modified table
  by the counterexamples.
  We can prove that a closed and consistent table builds a minimal
  finite automaton (omitted for space reasons).
  A new automaton $M'$ will be constructed when the extended table
  $\Otable$ is closed and consistent. Since $M'$ handles the
  counterexample, $M'$ has more equivalent states to the minimal
  automaton accepted the given language, compared with $M$. Repeating
  this process, the automaton associated with a closed and consistent
  table has the same number of the minimal automaton which accepted
  the given language. Since the minimal automaton is unique, the two
  minimal automata are equal. The teacher relies yes to the automaton
  at such a point. Therefore, with respect to the number of the states
  of the minimal automaton accepted the given language, equivalence
  queries are finite and the algorithm terminates finally.
\end{proof}

In the following, we analyse the number of queries and the execution
time of \nlstar\ in the worst case. Let $M$ be the minimal automaton
accepting the given language and let $M$ have $\mathfrak{s}$ states.
Let $\mathfrak{b}$ be a bound on the maximum length of the
counterexamples presented by the teacher.

From the Figure~\ref{LANA_a3} (line~\ref{codeline:init}), we know that $S$ and $E$ contain one element $\lambda$ initially. As the algorithm runs, it will add one element to $S$ when $\Otable$ is not closed (line~\ref{codeline:closedextension}). And it will add one element to $E$ when $\Otable$ is not consistent (line\ref{codeline:consistentextension}). For each counterexample of length at most $\mathfrak{b}$ presented by the teacher, the algorithm will add at most $\mathfrak{b}$ elements to $S$ (line~\ref{codeline:addcounterexample}).

Thus, the cardinality of $S$ is depends on $\mathfrak{s}$ and $\mathfrak{b}$. In detail, $S$ is at most $$1+(\mathfrak{s}-1)+\mathfrak{b}(\mathfrak{s}-1)=\mathfrak{s}+\mathfrak{b}(\mathfrak{s}-1)$$ because $\Otable$ can be not closed at most $\mathfrak{s}-1$ times. As the same as the teacher replies counterexamples at most $\mathfrak{s}-1$ times. And each time the teacher replies with a counterexample of length $\mathfrak{b}$, $S$ will be increased by at most $\mathfrak{b}$ elements. 

The cardinality of $E$ is at most $\mathfrak{s}$, because $\Otable$ can be not consistent at most $\mathfrak{s}-1$ times.   

The cardinality of $S\cdot \Alp_n$ could calculate from two parts: the cardinality of $S$ and the cardinality of $\Alp_n$. We already know the cardinality of $S$ is at most $\mathfrak{s}+\mathfrak{b}(\mathfrak{s}-1)$. As the definition of $\Alp_n$, let $\mathfrak{k}$ be the cardinality of $\Sigma$. Therefore the cardinality of $\Alp_n$ is $\mathfrak{k}+n+2$ and the cardinality of $S\cdot \Alp_n$ is $(\mathfrak{k}+n+2)(\mathfrak{s}+\mathfrak{b}(\mathfrak{s}-1))$ at most.

Therefore, the maximum cardinality of $(S\cup S\cdot \Alp_n)\cdot E$ is at most $$(\mathfrak{k}+n+2)(\mathfrak{s}+\mathfrak{b}(\mathfrak{s}-1))\mathfrak{s}=O((\mathfrak{k}+n)\mathfrak{b}\mathfrak{s}^2).$$ 
% The maximum length of any word in $(S\cup S\cdot \Alp_n)\cdot E$ is at most $$

The minimal automaton $M$ has $\mathfrak{s}$ states and the table $\Otable$ has one row initially. In the worst case, the table $\Otable$ adds only a distinguished row by every counterexample. The algorithm produces at most $\mathfrak{s}-1$  equivalence queries.

\paragraph{Running \nlstar: An Example}
Given a finite alphabet $\Sigma=\{a,b\}$, we have an example of learning a language $L$ representing as canonical nominal regular expression $cne=ab\langle 1^*\rangle$. 

In the first step, we initialize $S_1=\{\emptystring\}$, $E_1= \{\emptystring\}$, $n=0$ and ${\Alp_0}=\Sigma$, and construct $T_1$ as follows.\\
\textbf{Step 1}\\
\begin{minipage}{0.5\textwidth}
  \begin{center}
    $T_1$=
    \begin{tabular}{l|l|l}
      
      {$\reg{\_}$}  &{}  &      {$\emptystring$} \\
      \hline
      {0}  &  ${\emptystring}$ & {P} \\
      \hline
      {0}  & {$a$} & {P} \\
      {0}  &{$b$} & {0}\\
      
    \end{tabular}
  \end{center}
\end{minipage}
\begin{minipage}{0.5\textwidth}
  $(S_1,E_1,T_1,{\Alp_0})$ consistent?  There is only one row in $S$, thus the table is consistent.\\
  $(S_1,E_1,T_1,{\Alp_0})$ closed? No, $\row(b)\neq \row(\emptystring)$.\\
  So, $S_2\leftarrow S_1\cup \{b\}$ and we go to step 2.
\end{minipage}
\textbf{Step 2}\\
Let $S_2= S\cup \{b\}$ and $E_2=E$ and then construct a new observation table $(S_2,E_2,T_2,{\Alp_0})$ through membership queries.\\
\begin{minipage}{0.5\textwidth}
  \begin{center}
    $T_2$=
    \begin{tabular}{l|l|l}
      
      {$\reg{\_}$}  &{}  &      {$\emptystring$} \\
      \hline
      {0}  & ${\emptystring}$ & {P} \\
      {0}  &{$b$} & {0} \\
      \hline
      {0}  &{$a$} & {P}\\
      {0}  &{$aa$} & {0}\\
      {0}  &{$ab$} & {P}\\
    \end{tabular}
  \end{center}
\end{minipage}
\begin{minipage}{0.5\textwidth}
  $(S_2,E_2,T_2,{\Alp_0})$ closed? $\surd$ \\
  $(S_2,E_2,T_2,{\Alp_0})$ consistent? $\surd$\\
  Then, we compute the automaton $M$:
  
  \begin{tikzpicture}[>=stealth',shorten >=1pt,auto,node distance=15mm]
    \node[initial,state,inner sep=1pt,minimum size=0pt] (q0)      {$q_0$};
    \node[state,inner sep=1pt,minimum size=0pt]         (q1) [right of=q0]  {$q_1$};
    
    \path[->] (q0)  edge [bend left] node {$a$} (q1)
    edge [loop above] node {$b$} (q0)
    (q1) edge [bend left] node {$a$} (q0)
    edge [loop above] node {$b$} (q1);
    
  \end{tikzpicture}
  \\
  Teacher replies no and a counterexample, say, $ab\ll 1. \gg$. It is in $L$ not in $M$. And we go to step 3.
\end{minipage}
\textbf{Step 3}\\
Let $S_3\leftarrow S_2\cup \{a,ab, ab\ll 1.,ab\ll 1. \gg\}$, $E_3\leftarrow E_2$, and $n=1$, and then, the alphabet is extended to ${\Alp_1}=\Sigma\cup \underline{n}\cup\{\ll,\gg\}$. We should construct new observation tables sequentially through membership queries. Then we check the new table for closeness and consistency.\\
\begin{minipage}{0.5\textwidth}
  \begin{center}
    $T_3$=
    \begin{tabular}{l|l|l}
      
      {$\reg{\_}$}  &{}  &      {$\emptystring$}     \\
      \hline
      {0}  &${\emptystring}$ & {P} \\
      
      {0}  &{$b$} & {0} \\
      {0}  &{$a$} & {P}\\
      {0}  &{$ab$} & {P}\\
      {1}  & {$ab\ll 1.$} & {P} \\
      {0}  &{$ab\ll 1.\gg$} & {1} \\
      \hline
      {1}  &{$\ll 1.$} & {0} \\
      {0}  &{$a$} & {P} \\
      {0}  &{$ab\ll 1.\gg a$} & {0} \\
      {0}  &{$ab\ll 1.\gg b$} & {0} \\
      {1}  &{$ab\ll 1.1$} & {P} \\
      {1}  &{$ab\ll 1.a$} & {0} \\
      {$\cdots$}  &{$\cdots$} & {$\cdots$} \\
    \end{tabular}
  \end{center}
\end{minipage}
\begin{minipage}{0.5\textwidth}
  $(S_3,E_3,T_3,{\Alp}_1)$ consistent? \\
  No, $\row(ab)=\row(\emptystring)$ but $\row(ab\ll 1.)\neq \row(\ll 1.)$ . \\
  $(S_3,E_3,T_3,{\Alp}_1)$ closed? \\
   No,  $\row(\ll 1.)$ with $\reg{\ll 1.}=1$ has a fresh content. \\
\end{minipage}
\textbf{Step 4}\\
Let $S_4\leftarrow S_3\cup \{\ll 1. \}$,$E_4\leftarrow E_3\cup \{\ll 1.\}$, we should construct a new observation table $(S_4,E_4,T_4,{\Alp}_1)$ and check the new table for closeness and consistency.\\
\begin{minipage}{0.5\textwidth}
  \begin{center}
    $T_4$=
     \begin{tabular}{l|l|l|l}
    
    {$\reg{\_}$}  &{}  &      {$\emptystring$}   &      {$\ll 1.$} \\
    \hline
    {0}  &${\emptystring}$ & {P} & {0}\\
    
    {0}  &{$b$} & {0} & {0}\\
    {0}  &{$a$} & {P}& {0}\\
    {0}  &{$ab$} & {P}& {P}\\
    {1}  &{$ab\ll 1.$} & {P} & {$\bot$}\\
    {0}  & {$ab\ll 1.\gg$} & {1} & {0}\\	
    {1}  &{$\ll 1.$} & {0}& {$\bot$}\\
    \hline
    {0}  &	{$ba$} & {0}& {0}\\
    {0}  &	{$bb$} & {0}& {0}\\
    {1}  &	{$ab\ll 1.1$} & {P} & {$\bot$}\\
    {1}  &{$ab\ll 1.a$} & {0}  & {$\bot$}\\
    {$\cdots$}  &{$\cdots$} & {$\cdots$} & {$\cdots$}\\
  \end{tabular}
  \end{center}
\end{minipage}
\begin{minipage}{0.5\textwidth}
Once the table is closed and consistent, we ask an equivalence query.\\
Finally, the teacher replies \quo{yes} to an equivalence query about. The learning progress terminates. The learner automaton is as below.\\
\end{minipage}
\begin{tikzpicture}[>=stealth',shorten >=1pt,auto,node distance=15mm]
  \node[initial,state,inner sep=1pt,minimum size=0pt] (q0)      {$q_0$};
  \node[state,inner sep=1pt,minimum size=0pt]         (q1) [right of=q0]  {$q_1$};
  \node[state,inner sep=1pt,minimum size=0pt]         (q2) [right of=q1]  {$q_2$};
  \node[state,inner sep=1pt,minimum size=0pt]         (q3) [above of=q2]  {$q_3$};
  \node[state,accepting,inner sep=1pt,minimum size=0pt]         (q4) [right of=q2]  {$q_4$};
  \node[state,inner sep=1pt,minimum size=0pt]         (q5) [right of=q4]  {$q_5$};
  \node[state,inner sep=1pt,minimum size=0pt]         (q6) [right of=q3]  {$q_6$};
  \path[->] (q0)  edge node {$a$} (q1)
  edge[ bend right]  node[below] {$b$} (q5)
  
  (q1) edge[ bend right] node[below] {$a$} (q5)
  edge node {$b$} (q2)
  (q2) edge[ bend right]  node[below] {$a,b$} (q5)
  edge [ bend left] node {$\ll$} (q3)
  (q3) edge node {$a,b$} (q6)
  edge [loop above] node {$1$} (q3)
  edge [above] node {$\gg$} (q4)
  (q4) edge [bend right] node[below] {$a,b$} (q5)
  (q5) edge [loop right] node {$a,b$} (q5)
  (q6) edge [loop above] node {$a,b,1$} (q6)
  edge [bend left] node {$\gg$} (q5);
\end{tikzpicture}

%%% Local Variables:
%%% mode: latex
%%% TeX-master: "main"
%%% End:

%% file: supp.tex
\subsection{Preliminaries}
Before introducing our learning algorithm, some auxiliary notions are
necessary to give a concrete representation of nominal languages
and automata.
In fact, binders yield infinitely many equivalent representation of nominal words
due to alpha-conversion.
For instance, $\ll \aname. \aname \gg$ and
$\ll \aname[m].\aname[m] \gg$ are the same nominal word \emph{up-to}
renaming of their bound name.
We introduce \emph{canonical expressions} to give a finitary
representation of nominal regular languages.

\begin{definition}[Canonical expressions]\label{def:CE}
  Let $1 \leq n  \in \Nat$ a natural number and $ne$ a closed nominal
  regular expression.
  The \emph{$n$-canonical representation} $\ce{n}{ne}$ of $ne$ is
  defined as follows
  \begin{itemize}
  \item $ne \in \{\emptyword,\emptyset\}\cup \Sigma \implies \ce{n}{ne} =ne$
  \item $\ce{n}{ne+ne'} =\ce{n}{ne}+\ce{n}{ne'}$
  \item $\ce{n}{ne\cdot ne'} =\ce{n}{ne}\cdot \ce{n}{ne'}$,
  \item $\ce{n}{ne^*} =(\ce{n}{ne})^*$
  \item $ne=\langle \aname. ne'\rangle \implies \ce{n}{ne} =\langle n. \ce{n+1}{ne'[n/\aname]}\rangle$
  \end{itemize}
  where $ne'[n/\aname]$ is the capture-avoiding substitution of
  $\aname$ for $n$ in $ne'$.
  The \emph{canonical representation} of $ne$ is the term
  $\ce 1 {ne}$.
\end{definition}
Note that the map $\ce{\_}{\_}$ does not change the structure of the
nominal regular expression $ne$.  Basically, $\ce \_ \_$ maps nominal
regular expressions to terms where names are concretely represented as
positive numbers.
\begin{example}
  Given $\Sigma=\{a,b\}$, we give some examples of canonical
  representations of nominal expressions.
  \begin{small}
  \begin{itemize}
  \item $aba$ is the canonical representations of itself; indeed $\ce{1}{aba} = aba$
  \item $\ce{1}{\langle \aname . a\, \aname \rangle} = \ce{1}{\langle \aname[m]. a\, \aname[m] \rangle}{}=\langle 1.a\,1 \rangle$
    is the canonical representation of both $\langle \aname.a\,\aname \rangle$ and $\langle \aname[m].a\,\aname[m] \rangle$
  \item the canonical representation of
    $ne = \langle \aname.a\,\aname \langle \aname[m].\aname\, b\, \aname[m]
    \rangle\rangle\langle \aname[m].\aname[m] \rangle$ is
    $\ce{1} {ne} = \langle 1.a1\ce{2}{(\langle \aname[m].1\, b \, \aname[m] \rangle)}\langle
    1.1\rangle \\=\langle 1.a\,1\langle 2.1\,b\,2 \rangle \rangle\langle 1.1
    \rangle$.
  \end{itemize}  
  \end{small}
  Note that the map $\ce{\_}{\_}$ replaces names with numbers so that
  alpha-equivalent expressions are mapped to the same term (second example above).
\end{example}
Canonical expressions are the linguistic counter part of the mechanism
used in the definition of $(\Sigma, \Nom_{fin})$-automaton give
in~\cite{Kurz0T13}, where transitions with indexes can consume
bound names of words.
Thus, we use canonical expressions in order to represent nominal
languages concretely.

%%% Local Variables:
%%% mode: latex
%%% TeX-master: "main"
%%% End:

%% file: conc.tex
The learning algorithm $\lstar$ was introduced more than thirty years
ago and has been intensively extended to many types of models in
following years. This algorithm continues to attract the attention of
many researchers~\cite{Angluin17,SchroderKMW17,Moerman17}.

We designed a learning algorithm for a class of languages over
infinite alphabet; more precisely, we have considered nominal regular
languages with binders~\cite{Kurz0T12,Kurz0T13}.
We have tackled the finitary representations of the alphabets, words
and automata for retaining the basic scheme and ideas of
$\lstar$. Hence, we revised and added definitions for the nominal
words and automata.
Further, accounting for names and the allocation and deallocation
operations, we revised the data structures and notions in
$\lstar$. Accordingly, we have proposed the learning algorithm,
$\nlstar$, to stress the progress of learning a nominal language with
binders. We have proved the correctness and analysed the complexities
of $\nlstar$.

As for \lstar, a key factor for the effectiveness of and \nlstar is
the selection of counterexamples.
Due to possibly infinite number of candidate counterexamples, the
selection of the counterexamples is non-deterministic in \nlstar.
We are developing an implementation of $\nlstar$ to study effective
mechanisms to resolve this non-determinism.
Interestingly, the rich structure of nominal automata offer different
directions to solve this problem.
In fact, we started to investigate this issue and defined two
different strategies used by the teacher to generate counterexamples
based on the \quo{size} of the counterexamples and on the preference
of the teacher for counterexamples with maximal or minimal number of
fresh names.
Initial experiments show how the strategies impact on the
convergence of \nlstar.
This immediately suggest that \nlstar could be improved by designing
different strategies to generate \quo{better} counterexamples, that is
counterexamples that allow the learner to learn \quo{more quickly}.
% \begin{new}
%   The use of symbol \quo{P} allows us to attain a significant
%   improvement on the performance of the learning process (as hinted
%   from Table~\ref{tb:comparison}). We are now trying to find out the
%   exact factors which have been affected. In the future, we conclude a
%   new strategy within the \quo{P} symbol.
% \end{new}
%%% Local Variables:
%%% mode: latex
%%% TeX-master: "main"
%%% End: